\documentclass[a4paper]{article}

\usepackage{amsmath}
\usepackage{amssymb}
\usepackage{amsthm}
\usepackage[dvipdf]{graphicx}

\newtheorem{theorem}{Theorem}
\newtheorem{lemma}{Lemma}
\newtheorem{corollary}{Corollary}

 \usepackage{color}

\newcommand{\YES}{\mathsf{yes}}
\newcommand{\NO}{\mathsf{no}}

\newcounter{one}
\newcounter{two}
\newcounter{three}
\newenvironment{listing}[1]{%
    \begin{list}{*}{%
    \settowidth{\labelwidth}{#1}%
    \setlength{\leftmargin}{\labelwidth}%
    \advance \leftmargin by 12pt
    \setlength{\itemsep}{0pt}%
    \setlength{\parsep}{0pt}%
    \setlength{\topsep}{0pt}%
    \setlength{\parskip}{0pt}%
    }%
    }{%
    \end{list}}

\newcommand{\seq}[1]{\langle #1 \rangle}
\newcommand{\dprime}{{\prime\prime}}

\newcommand{\pro}{\Pi}
\newcommand{\rrule}{\mathcal{A}}

\begin{document}
	\title{Reconfiguring spanning and induced subgraphs\thanks{This work is partially supported by JST ERATO Grant Number JPMJER1201, JST CREST Grant Number JPMJCR1402, and JSPS KAKENHI Grant Numbers JP16K00004 and JP17K12636, Japan. Research by Canadian authors is supported by the Natural Science and Engineering Research Council of Canada.}}

\author{%
	Tesshu Hanaka\footnote{Kyushu University, Fukuoka, Japan, \texttt{3EC15004S@s.kyushu-u.ac.jp}}\and
	Takehiro Ito\footnote{Tohoku University, Sendai, Japan, \texttt{takehiro@ecei.tohoku.ac.jp, haruka.mizuta.s4@dc.tohoku.ac.jp, a.suzuki@ecei.tohoku.ac.jp}}\and
	Haruka Mizuta\footnotemark[3] \and
	Benjamin Moore\footnote{University of Waterloo, Waterloo, Canada, \texttt{\{brmoore, nishi, v7subram, kvaidyan\}@uwaterloo.ca}} \and
	Naomi Nishimura\footnotemark[4] \and
	Vijay Subramanya\footnotemark[4] \and
	Akira Suzuki\footnotemark[3]\and
	Krishna Vaidyanathan\footnotemark[4]
}


\maketitle

\begin{abstract}
{\sc Subgraph reconfiguration} is a family of problems focusing on the reachability of the solution space in which feasible solutions are subgraphs, represented either as sets of vertices or sets of edges, satisfying a prescribed graph structure property. 
%
%
%
Although there has been previous work that can be categorized as {\sc subgraph reconfiguration}, most of the related results appear under the name of the property under consideration; for example, independent set, clique, and matching.
In this paper, we systematically clarify the complexity status of {\sc subgraph reconfiguration} with respect to graph structure properties. 
%
%
%

%
\end{abstract}

\section{Introduction}
	Combinatorial reconfiguration~\cite{IDHPSUU11,H13,N17} studies the reachability/connectivity of the solution space formed by feasible solutions of an instance of a search problem. 
	More specifically, consider a graph such that each node in the graph represents a feasible solution to an instance of a search problem $P$, and there is an edge between nodes representing any two feasible solutions that are ``adjacent,'' according to a prescribed {\em reconfiguration rule} $\rrule$; such a graph is called the {\em reconfiguration graph} for $P$ and $\rrule$.
	In the {\em reachability problem} for $P$ and $\rrule$, we are given {\em source} and {\em target} solutions to $P$, and the goal is to determine whether or not there is a path between the two corresponding nodes in the reconfiguration graph for $P$ and $\rrule$.
	We call a desired path a {\em  reconfiguration sequence} between source and target solutions, where a  {\em reconfiguration step} from one solution to another corresponds to an edge in the path.

	\subsection{Subgraph reconfiguration}
	In this paper, we use the term {\sc subgraph reconfiguration} to describe a family of reachability problems that take subgraphs (more accurately, vertex subsets or edge subsets of a given graph) as feasible solutions.
	Each of the individual problems in the family can be defined by specifying the node set and the edge set of a reconfiguration graph, as follows. 
	(We use the terms {\em node} for reconfiguration graphs and {\em vertex} for input graphs.)
		\medskip

		\begin{table}[t]
			\begin{center}
				\caption{Subgraph representations and variants}
				\begin{tabular}{|c|c|l|}
					\hline
					~Subgraph representations~ & ~Variant names~ & ~Known reachability problems~ \\
					\hline
					\hline
					~edge subset~ & ~edge~ & ~spanning tree \cite{IDHPSUU11}~ \\
					& & ~matching \cite{IDHPSUU11,M15}, and $b$-matching \cite{M15}~ \\
					\hline
					& & ~clique \cite{IOO15}~ \\
					& & ~independent set \cite{IDHPSUU11,KMM12}~ \\
					~vertex subset~ & ~induced~ & ~induced forest \cite{MNRSS17}~ \\
					& & ~induced bipartite \cite{MNRSS17}~ \\
					& & ~induced tree \cite{WYA16}~ \\
					\cline{2-3}
					& ~spanning~ & ~clique \cite{IOO15}~ \\
					\hline
				\end{tabular}
				\label{tab:variant}
			\end{center}
		\end{table}
		
		\begin{figure}[t]
			\centering
			\includegraphics[width=0.8\linewidth]{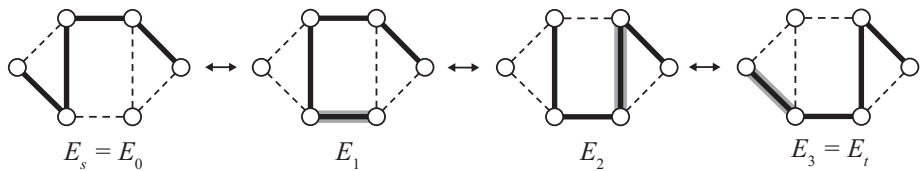}
			\vspace{-1em}
			\caption{A reconfiguration sequence $\seq{E_0,E_1,E_2,E_3}$ in the edge variant under the TJ rule (also under the TS rule) with the property ``a graph is a path,'' where the edges forming solutions are depicted by thick lines.}
			\label{fig:edge}
		\end{figure}

		\noindent
		{\bf Nodes of a reconfiguration graph}.
		The set of feasible solutions (i.e., subgraphs) can be defined in terms of a specified graph structure property $\pro$ which subgraphs must satisfy; for example, ``a graph is a tree,'' ``a graph is edgeless (an independent set),'' and so on.
		By the choice of how to represent subgraphs, each specific problem in the family can be categorized into one of three variants. (See also \tablename~\ref{tab:variant}.)
		If a subgraph is represented as an edge subset, which we will call the {\em edge variant}, then the subgraph formed (induced) by the edge subset must satisfy $\pro$.
		For example, \figurename~\ref{fig:edge} illustrates four subgraphs represented as edge subsets, where $\pro$ is ``a graph is a path.''
		On the other hand, if a subgraph is represented as a vertex subset, we can opt either to require that the subgraph induced by the vertex subset satisfies $\pro$ or that the subgraph induced by the vertex subset contains at least one spanning subgraph that satisfies $\pro$; we will refer to these as the {\em induced variant} and {\em spanning variant}, respectively.
		For example, if $\pro$ is ``a graph is a path,'' then in the induced variant, the vertex subset must induce a path, whereas in the spanning variant, the vertex subset is feasible if its induced subgraph contains at least one Hamiltonian path.
		Figure~\ref{fig:ind_spa} illustrates feasible vertex subsets of the induced variant and spanning variant.
		In the figure, the vertex subset $V^\prime_1$ is feasible in the spanning variant, but is not feasible in the induced variant, because it contains a spanning path but does not induce a path.
		As can be seen by this simple example, in the spanning variant, we need to pay attention to the additional complexity of finding a spanning subgraph and the complications resulting from the fact that the subgraph induced by the vertex subset may contain more than one spanning subgraph which satisfies $\pro$.
		\medskip
		
		\noindent
		{\bf Edges of a reconfiguration graph}.
		Since we represent a feasible solution by a set of vertices (or edges) in any variant, we can consider that tokens are placed on each vertex (resp., edge) in the feasible solution.
		Then, in this paper, we mainly deal with the two well-known reconfiguration rules, called the token-jumping (TJ)~\cite{KMM12} and token-sliding (TS) rules~\cite{BC09,HD05,KMM12}.  
		In the former, a token can move to any other vertex (edge) in a given graph, whereas in the latter it can move only to an adjacent vertex (adjacent edge, that is sharing a common vertex.)
		For example, \figurename~\ref{fig:edge} and \figurename~\ref{fig:ind_spa} illustrate reconfiguration sequences under the TJ rule for each variant.
		Note that the sequence in \figurename~\ref{fig:edge} can also be considered as a sequence under the TS rule.
		In the reconfiguration graph, two nodes are adjacent if and only if one of the two corresponding solutions can be obtained from the other one by a single move of one token that follows the specified reconfiguration rule. 
		Therefore, all nodes in a connected component of the reconfiguration graph represent subgraphs having the same number of vertices (edges).
		
		We note in passing that since in most cases we wish to retain the same number of vertices and/or edges, we rarely use the token-addition-and-removal (TAR) rule~\cite{IDHPSUU11,KMM12}, where we can add or remove a single token at a time, for {\sc subgraph reconfiguration} problems.

		\begin{figure}[t]
			\centering
			\includegraphics[width=0.75\linewidth]{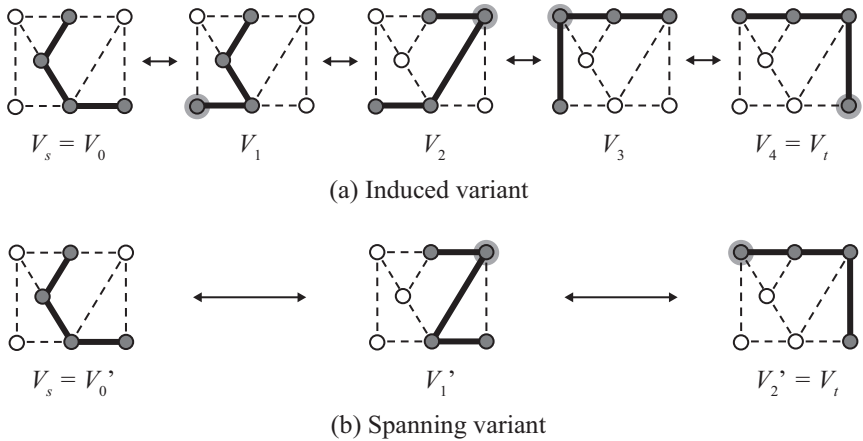}
			\vspace{-1em}
			\caption{Reconfiguration sequences $\seq{V_0,V_1,V_2,V_3,V_4}$ in the induced variant under the TJ rule and $\seq{V^\prime_0,V^\prime_1,V^\prime_2}$ in the spanning variant under the TJ rule with the property ``a graph is a path,'' where the vertices forming solutions are depicted by colored circles, and the subgraphs satisfying the property by thick lines.}
			\label{fig:ind_spa}
		\end{figure}
	
	\subsection{Previous work}
		Although there has been previous work that can be categorized as {\sc subgraph reconfiguration}, most of the related results appear under the name of the property $\pro$ under consideration.
		Accordingly, we can view reconfiguration of independent sets~\cite{IDHPSUU11,KMM12} as the induced variant of {\sc subgraph reconfiguration} such that the property $\pro$ is ``a graph is edgeless.''  
		Other examples can be found in \tablename~\ref{tab:variant}.
		We here explain only known results which are directly related to our contributions. 
		
		Reconfiguration of cliques can be seen as both the spanning and the induced variant; the problem is PSPACE-complete under any rule, even when restricted to perfect graphs~\cite{IOO15}.
		Indeed, for this problem, the rules TAR, TJ, and TS have all been shown to be equivalent from the viewpoint of polynomial-time solvability. 
		It is also known that reconfiguration of cliques can be solved in polynomial time for several well-known graph classes~\cite{IOO15}.
%
		
		
		Wasa et al.~\cite{WYA16} considered the induced variant under the TJ and TS rules with the property $\pro$ being ``a graph is a tree.'' 
		They showed that this variant under each of the TJ and TS rules is PSPACE-complete, and is W[1]-hard when parameterized by both the size of a solution and the length of a reconfiguration sequence.
		They also gave a fixed-parameter algorithm when parameterized by both the size of a solution and the maximum degree of an input graph, under both the TJ and TS rules.
		In closely related work, Mouawad et al.~\cite{MNRSS17} considered the induced variants of {\sc subgraph reconfiguration} under the TAR rule with the properties $\pro$ being either ``a graph is a forest'' or ``a graph is bipartite.''
		They showed that these variants are W[1]-hard when parameterized by the size of a solution plus the length of a reconfiguration sequence.

	\subsection{Our contributions}
		In this paper, we study the complexity of {\sc subgraph reconfiguration} under the TJ and TS rules.
		(Our results are summarized in Table~\ref{tab:result}, together with known results, where an {\em $(i,j)$-biclique} is a complete bipartite graph with the bipartition of $i$ vertices and $j$ vertices.) 
		As mentioned above, because we consider the TJ and TS rules, it suffices to deal with subgraphs having the same number of vertices or edges. 
		Subgraphs of the same size may be isomorphic for certain properties $\pro$, such as  ``a graph is a path'' and ``a graph is a clique,''
                because there is only one choice of a path or a clique of a particular size.  
		On the other hand, for the property ``a graph is a tree,'' there are several choices of trees of a particular size.
		(We will show an example in Section~\ref{sec:edge} with \figurename~\ref{fig:edge_tree}.)
		
 		As shown in Table 2, we systematically clarify the complexity of {\sc subgraph reconfiguration} for several fundamental graph properties.
		In particular, we show that the edge variant under the TJ rule is computationally intractable for the property ``a graph is a path'' but 
tractable for the property ``a graph is a tree.''
		This implies that the computational (in)tractability does not follow directly from the inclusion relationship of graph classes required as the properties $\pro$;
one possible explanation is that the path property implies a specific graph, whereas the tree property allows several choices of trees, making the problem easier.

		\begin{table}[t]
			\begin{center}
				\caption{Previous and new results}
				\begin{tabular}{|c||c|c|c|}
					\hline
					{\bf Property $\pro$} & {\bf Edge}  & {\bf Induced} & {\bf Spanning} \\
					 & {\bf variant} & {\bf variant} & {\bf variant} \\
					\hline
					\hline
					{\bf path} & NP-hard (TJ) & PSPACE-c. (TJ, TS) & PSPACE-c. (TJ, TS) \\
					& [Theorem~\ref{the:edgepath}] & [Theorems~\ref{the:path_tj},~\ref{the:pathcycle_ts}] & [Theorems~\ref{the:path_tj},~\ref{the:pathcycle_ts}]  \\
					\hline
					{\bf cycle} & P (TJ, TS) & PSPACE-c. (TJ, TS) & PSPACE-c. (TJ, TS)  \\
					& [Theorem~\ref{the:edgecycle}] & [Theorems~\ref{the:cycle_tj},~\ref{the:pathcycle_ts}] & [Theorems~\ref{the:cycle_tj},~\ref{the:pathcycle_ts}] \\ 
					\hline
					{\bf tree} & P (TJ)  & ~PSPACE-c. (TJ, TS)~ & P (TJ)  \\
					&  [Theorem~\ref{the:edgetree}] & \cite{WYA16} & PSPACE-c. (TS) \\
					& & & [Theorems~\ref{the:spanningtreetj},~\ref{the:spanningtreets}] \\ 
					\hline
					{\bf $(i,j)$-biclique} & P (TJ, TS) & PSPACE-c. for $i=j$ (TJ)~ & ~NP-hard for $i=j$ (TJ)~\\
					& [Theorem~\ref{the:edgebiclique}] & PSPACE-c. for fixed $i$ (TJ) & P for fixed $i$ (TJ) \\
					& & [Corollary~\ref{col:inducedbiclique},~Theorem~\ref{the:inducedbiclique}] & [Theorems~\ref{the:spanningbiclique},~\ref{the:spanningbicliquesolve}] \\
					\hline
					{\bf clique} & P (TJ, TS) & PSPACE-c. (TJ, TS) & PSPACE-c. (TJ, TS) \\
					& [Theorem~\ref{the:edgeclique}]  & \cite{IOO15} & \cite{IOO15}\\ 
					\hline
					~{\bf diameter}~ & & PSPACE-c. (TS) & PSPACE-c. (TS)  \\
					~{\bf two}~ &  & [Theorem~\ref{the:diamtwo}] & [Theorem~\ref{the:diamtwo}]  \\
					\hline      
					\hline  
					{\bf any} & XP for solution & XP for solution & ~XP for solution~ \\
					{\bf property}  &  size (TJ, TS) & size (TJ, TS) & size (TJ, TS) \\
					& [Theorem~\ref{the:xp}] & [Theorem~\ref{the:xp}] & [Theorem~\ref{the:xp}] \\
					\hline
				\end{tabular}
				\label{tab:result}
			\end{center}
		\end{table}


\subsection{Preliminaries}
	Although we assume throughout the paper that an input graph $G$ is simple, all our algorithms can be easily extended to graphs having multiple edges. 
	We denote by $(G, V_s, V_t)$ an instance of a spanning variant or an induced variant whose input graph is $G$ and source and target solutions are vertex subsets $V_s$ and $V_t$ of $G$. 
	Similarly, we denote by $(G, E_s, E_t)$ an instance of the edge variant.
	We may assume without loss of generality that $|V_s| = |V_t|$ holds for the spanning and induced variants, and $|E_s| = |E_t|$ holds for the edge variant; otherwise, the answer is clearly $\NO$ since under both the TJ and TS rules, all solutions must be of the same size.

\section{General algorithm} \label{sec:general}
	In this section, we give a general XP algorithm when the size of a solution (that is, the size of a vertex or edge subset that represents a subgraph) is taken as the parameter.
	For notational convenience, we simply use {\em element} to represent a vertex (or an edge) for the spanning and induced variants (resp., the edge variant), and {\em candidate} to represent a set of elements (which does not necessarily satisfy the property $\pro$.) 
	Furthermore, we define the {\em size} of a given graph as the number of elements in the graph.
%
%
	\begin{theorem}\label{the:xp}
		Let $\pro$ be any graph structure property, and let $f(k)$ denote the time to check if a candidate of size $k$ satisfies $\pro$.
		Then, all of the spanning, induced, and edge variants under the TJ or TS rules can be solved in time $O(n^{k+1}k+n^k f(k))$, where $n$ is the size of a given graph and $k$ is the size of a source {\rm (}and target{\rm )} solution.
		Furthermore, a shortest reconfiguration sequence between source and target solutions can be found in the same time bound, if it exists.
	\end{theorem}
	\begin{proof}
		Our claim is that the reconfiguration graph can be constructed in the stated time.
		Since a given source solution is of size $k$, it suffices to deal only with candidates of size exactly $k$.
		For a given graph, the total number of possible candidates of size $k$ is $O(n^k)$.
		For each candidate, we can check in time $f(k)$ whether it satisfies $\pro$.
		Therefore, we can construct the node set of the reconfiguration graph in time $O(n^{k}f(k))$.
		We then obtain the edge set of the reconfiguration graph.
		Notice that each node in the reconfiguration graph has $O(nk)$ adjacent nodes, because we can replace only a single element at a time. 
		Therefore, we can find all pairs of adjacent nodes in time $O(n^{k+1}k)$.
		
		In this way, we can construct the reconfiguration graph in time $O(n^{k+1}k+n^k f(k))$ in total.
		The reconfiguration graph consists of $O(n^k)$ nodes and $O(n^{k+1}k)$ edges.
		Therefore, by breadth-first search starting from the node representing a given source solution, we can determine in time $O(n^{k+1}k)$ whether or not there exists a reconfiguration sequence between two nodes representing the source and target solutions.
		Notice that if a desired reconfiguration sequence exists, then the breadth-first search finds a shortest one.  
	\end{proof}

\section{Edge variants} \label{sec:edge}
	In this section, we study the edge variant of {\sc subgraph reconfiguration} for the properties of being paths, cycles, cliques, bicliques, and trees.
	
		We first consider the property  ``a graph is a path'' under the TJ rule. 
		\begin{theorem}\label{the:edgepath}
			The edge variant of {\sc subgraph reconfiguration} under the TJ rule is NP-hard for the property ``a graph is a path.''
		\end{theorem} 
		\begin{proof}
			We give a polynomial-time reduction from the {\sc Hamiltonian path} problem.
			Recall that a {\em Hamiltonian path} in a graph $G$ is a path that visits each vertex of $G$ exactly once.
			Given a graph $G$ and two vertices $s,t \in V(G)$ of $G$, the NP-complete problem {\sc Hamiltonian path} is to determine whether or not $G$ has a Hamiltonian path which starts from $s$ and ends in $t$~\cite{GJ79}.
			
			\begin{figure}[t]
				\centering
				\includegraphics[width=0.8\linewidth]{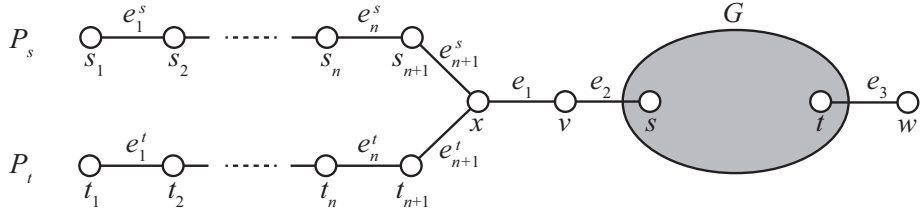}
				\caption{Reduction to the edge variant under the TJ rule for the property ``a graph is a path.''}
				\label{fig:edge_path}
			\end{figure}
			
			For an instance $(G,s,t)$ of {\sc Hamiltonian path}, we construct a corresponding instance $(G^\prime,E_s,E_t)$ of our problem, as follows.
			(See also \figurename~\ref{fig:edge_path}.)
			Let $n=|V(G)|$.
			We first add two new vertices $v$ and $x$ to $G$ with two new edges $e_1=xv$ and $e_2=vs$.  
			We then add two paths $P_s = \seq{s_1,s_2,\ldots,s_{n+1},x}$ and $P_t=\seq{t_1,t_2,\ldots,t_{n+1},x}$, where $s_1, s_2, \ldots,s_{n+1}$ and $t_1, t_2, \ldots,t_{n+1}$ are distinct new vertices.
			Each of $P_s$ and $P_t$ consists of $n+1$ edges; we denote by $e^s_1,e^s_2,\ldots,e^s_{n+1}$ the edges $s_1s_2,\ s_2s_3,\ \ldots, s_{n+1}x$ in $P_s$, respectively, and by $e^t_1,e^t_2,\ldots,e^t_{n+1}$ the edges $t_1t_2,\ t_2t_3,\ \ldots,t_{n+1}x$ in $P_t$, respectively.
			We finally add a new vertex $w$ with an edge $e_3 = tw$, completing the construction of $G^\prime$.
			We then set $E_s = \{ e^s_1,e^s_2,\ldots,e^s_{n+1},e_1,e_2\}$ and $E_t = \{e^t_1,e^t_2,\ldots,e^t_{n+1},e_1,e_2\}$; these edge subsets clearly form paths in $G^\prime$.
			We have thus constructed our corresponding instance $(G^\prime,E_s,E_t)$ in polynomial time.
			
			We now prove that an instance $(G,s,t)$ of {\sc Hamiltonian path} is a $\YES$-instance if and only if the corresponding instance $(G^\prime,E_s,E_t)$
			is a $\YES$-instance.  
			
			To prove the only-if direction, we first suppose that $G$ has a Hamiltonian path $P$ starting from $s$ and ending in $t$.
			Then, we construct an actual reconfiguration sequence from $E_s$ to $E_t$ using the edges in $P$.
			Notice that $P$ consists of $n-1$ edges.
			Thus, we first move the $n-1$ edges $e^s_1,e^s_2,\ldots,e^s_{n-1}$ in $E_s$ to the edges in $P$ one by one, and then move $e^s_n$ to $e_3$. 
			Next, we move $e^s_{n+1}$ to $e^t_{n+1}$, and then move the edges in $E(P) \cup \{e_3\}$ to $e^t_{n},e^t_{n-1},\ldots,e^t_1$ one by one.
			By the construction of $G^\prime$, we know that each of the intermediate edge subsets forms a path in $G^\prime$, as required.
			
			We now prove the if direction by 
			supposing that there exists a reconfiguration sequence $\seq{E_s=E_0,E_1,\ldots,E_\ell=E_t}$.
			Let $E_q$ be the first edge subset in the sequence such that $E(P_t) \cap E_q \neq \emptyset$;
			we claim that $E_q$ contains a Hamiltonian path in $G$. 
			First, notice that the edge in $E(P_t) \cap E_q$ is $e^t_{n+1}$; otherwise the subgraph formed by $E_q$ is disconnected.
			Since $|E_q|=|E_s|=n+3$ and $|E(P_s)| = n+1$, we can observe that $E_q$ contains no edge in $P_s$; otherwise the degree of $x$ would be three, or $E_q$ would form a disconnected subgraph.
			Therefore, the $n+2$ edges in $E_q \setminus \{ e^t_{n+1} \}$ must be chosen from $E(G) \cup \{ e_1,e_2,e_3 \}$.
			Since $|V(G)|=n$ and $E_q$ must form a path in $G^\prime$, we know that $E_q \setminus \{ e^t_{n+1} \}$ consists of $e_1,e_2,e_3$ and $n-1$ edges in $G$.
			Thus, $E_q \setminus \{ e^t_{n+1}, e_1, e_2, e_3 \}$ forms a Hamiltonian path in $G$ starting from $s$ and ending in $t$, as required.
		\end{proof}

		We now consider the property ``a graph is a cycle,'' as follows. 
		\begin{theorem}\label{the:edgecycle}
			The edge variant of {\sc subgraph reconfiguration} under each of the TJ and TS rules can be solved in linear time for the property ``a graph is a cycle.''
		\end{theorem}
		\begin{proof}
			Let $(G,E_s,E_t)$ be a given instance.
			We claim that the reconfiguration graph is edgeless, in other words, no feasible solution can be transformed at all. 
			Then, the answer is $\YES$ if and only if $E_s = E_t$ holds; this condition can be checked in linear time.

			Let $E'$ be any feasible solution of $G$, and consider a replacement of an edge $e^- \in E'$ with an edge $e^+$ other than $e^-$. 
			Let $u,v$ be the endpoints of $e^-$.
			When we remove $e^-$ from $E'$, the resulting edge subset $E' \setminus \{e\}$ forms a path whose ends are $u$ and $v$.
			Then, to ensure that the candidate forms a cycle, we can choose only $e^- = uv$ as $e^+$. 
			This contradicts the assumption that $e^+ \neq e^-$. 
		\end{proof}
	
		The same arguments hold for the property ``a graph is a clique,'' and we obtain the following theorem. 
		We note that, for this property, both induced and spanning variants (i.e., when solutions are represented by vertex subsets) are PSPACE-complete under any rule~\cite{IOO15}.
		\begin{theorem}\label{the:edgeclique}
			The edge variant of {\sc subgraph reconfiguration} under each of the TJ and TS rules can be solved in linear time for the property ``a graph is a clique.''
		\end{theorem}
%

		We next consider the property ``a graph is an $(i,j)$-biclique,'' as follows. 
		\begin{theorem}\label{the:edgebiclique}
		 	The edge variant of {\sc subgraph reconfiguration} under each of the TJ and TS rules can be solved in polynomial time for the property ``a graph is an $(i,j)$-biclique'' for any pair of positive integers $i$ and $j$. 
		\end{theorem}
		\begin{proof}
		We may assume without loss of generality that $i \le j$ holds. 
		We prove the theorem in the following three cases: 
{\bf Case 1:} $i = 1$ and $j \leq 2$;
{\bf Case 2:} $i,j \geq 2$; and
{\bf Case 3:} $i = 1$ and $j \geq 3$. 

		We first consider {\bf Case 1}, which is the easiest case.
		In this case, any $(1,j)$-biclique has at most two edges. 
		Therefore, by Theorem~\ref{the:xp} we can conclude that this case is solvable in polynomial time.
		
		We then consider {\bf Case 2}.
		We show that $(G,E_s,E_t)$ is a $\YES$-instance if and only if $E_s=E_t$ holds.
		To do so, we claim that the reconfiguration graph is edgeless, in other words, no feasible solution can be transformed at all. 
		To see this, because $i, j \ge 2$, notice that the removal of any edge $e$ in an $(i,j)$-biclique results in a bipartite graph with the same bipartition of $i$ vertices and $j$ vertices. 
		Therefore, to obtain an $(i,j)$-biclique by adding a single edge, we must add back the same edge $e$. 
	
		We finally deal with {\bf Case 3}.
		Notice that a $(1,j)$-biclique is a star with $j$ leaves, and its center vertex is of degree $j \ge 3$.
		Then, we claim that $(G,E_s,E_t)$ is a $\YES$-instance if and only if the center vertices of stars represented by $E_s$ and $E_t$ are the same.
			The if direction clearly holds, because we can always move edges in $E_s \setminus E_t$ into ones in $E_t \setminus E_s$ one by one.
			We thus prove the only-if direction; indeed, we prove the contrapositive, that is, the answer is $\NO$ if the center vertices of stars represented by $E_s$ and $E_t$ are different.
			Consider such a star $T_s$ formed by $E_s$. 
			Since $T_s$ has $j \ge 3$ leaves, the removal of any edge in $E_s$ results in a star having $j-1 \ge 2$ leaves. 
			Therefore, to ensure that each intermediate solution is a star with $j$ leaves, we can add only an edge of $G$ which is incident to the center of $T_s$.
			Thus, we cannot change the center vertex.
%
		\end{proof}
	
		In this way, we have proved Theorem~\ref{the:edgebiclique}.
		Although {\bf Case 1} takes non-linear time, our algorithm can be easily improved (without using Theorem~\ref{the:xp}) so that it runs in linear time.

		\begin{figure}[t]
			\centering
			\includegraphics[width=0.65\linewidth]{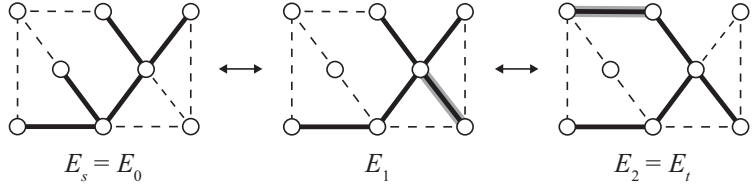}
			\caption{Reconfiguration sequence $\seq{E_0, E_1, E_2}$ in the edge variant under the TJ rule with the property ``a graph is a tree.''}
			\label{fig:edge_tree}
		\end{figure}
	
		We finally consider the property ``a graph is a tree'' under the TJ rule. 
		As we have mentioned in the introduction, for this property, there are several choices of trees even of a particular size, and a reconfiguration sequence does not necessarily consist of isomorphic trees  (see \figurename~\ref{fig:edge_tree}).  
		This ``flexibility'' of subgraphs may yield the contrast between Theorem~\ref{the:edgepath} for the path property and the following theorem for  the tree property.
		\begin{theorem} \label{the:edgetree}
			The edge variant of {\sc subgraph reconfiguration} under the TJ rule can be solved in linear time for the property ``a graph is a tree.''
		\end{theorem}
		\begin{proof}
			Suppose that $(G,E_s,E_t)$ is a given instance. 
			We may assume without loss of generality that $|E_s| = |E_t| \geq 2$; 
			otherwise $|E_s|=|E_t| \leq 1$ holds, and hence the instance is trivially a $\YES$-instance.
			We will prove below that any instance with $|E_s| = |E_t| \geq 2$ is a $\YES$-instance if and only if all the edges in $E_s$ and $E_t$ are contained in the same connected component of $G$.
			Note that this condition can be checked in linear time. 
			
			We first prove the only-if direction of our claim. 
			Since $|E_s| = |E_t| \geq 2$ and subgraphs always must retain a tree structure (more specifically, they must be connected graphs), observe that we can exchange edges only in the same connected component of $G$. 
			Thus, the only-if direction follows.

			To complete the proof, it suffices to prove the if direction of our claim.
			For notational convenience, for any feasible solution $E_i$ we denote by $T_i$ the tree represented by $E_i$, and by $V_i$ the vertex set of $T_i$.
			In this direction, we consider the following two cases: (a) $V_s \cap V_t = \emptyset$, and (b) $V_s \cap V_t \neq \emptyset$.
			
			We consider case (a), that is, $V_s \cap V_t = \emptyset$.
			Since $T_s$ and $T_t$ are contained in one connected component of $G$, there exists a path $\seq{v_0,v_1,\ldots,v_\ell}$ in $G$ such that $v_0 \in V_s$, $v_\ell \in V_t$ and $v_i \notin V_s \cup V_t$ for all $i \in \{1,2,\ldots,\ell-1 \}$. 
			Since $T_s$ is a tree, it has at least two degree-one vertices. 
			Let $v_s$ be any degree-one vertex in $V_s \setminus \{ v_0 \}$, and let $e_s$ be the leaf edge of $T_s$ incident to $v_s$.
			Then, we can exchange $e_s$ with $v_0 v_1$, and obtain another tree represented by the resulting edge subset $(E_s \cup \{ v_0 v_1 \}) \setminus \{e_s \}$.
			By repeatedly applying this operation along the path $\seq{v_1,v_2,\ldots,v_\ell}$, we can obtain a solution $E_k$ such that $V_k \cap V_t = \{ v_\ell \} \neq \emptyset$; 
			this case will be considered below. 
			
			We finally consider case (b), that is, $V_s \cap V_t \neq \emptyset$.
			Consider the graph $(V_s \cap V_t, E_s \cap E_t)$. 
			Then, $(V_s \cap V_t, E_s \cap E_t)$ is a forest, and let $G^\prime = (V^\prime,E^\prime)$ be a connected component (i.e., a tree) of $(V_s \cap V_t, E_s \cap E_t)$ whose edge set is of maximum size.
			We now prove that there is a reconfiguration sequence between $E_s$ and $E_t$ by induction on $k = |E_s \setminus E^\prime| = |E_t \setminus E^\prime|$.
			If $k = 0$, then $E_s = E^\prime = E_t$ and hence the claim holds.
			We thus consider the case where $k > 0$ holds.
			Since $G^\prime$ is a proper subtree of $T_t$, there exists at least one edge $e_t$ in $E_t \setminus E^\prime$ such that one endpoint of $e_t$ is contained in $V^\prime$ and the other is not.
			We claim that there exists an edge $e_s$ in $E_s \setminus E^\prime$ which can be moved into $e_t$, that is, the subgraph represented by the resulting edge subset $(E_s \cup \{ e_t \}) \setminus \{ e_s \}$ forms a tree. 
			If both endpoints of $e_t$ are contained in $V_s$ (not just $V^\prime$), $E_s \cup \{ e_t \}$ contains a cycle; let $C \subseteq E_s \cup \{ e_t \}$ be the edge set of the cycle.
			Since the subgraph $T_t$ has no cycle, there exists at least one edge in $C \setminus E_t$, and we choose one of them as $e_s$.
			On the other hand, if just one endpoint of $e_t$ is contained in $V_s$, then we choose a leaf edge of $T_s$ in $E_s \setminus E^\prime$ as $e_s$.
			Note that there exists such a leaf edge since $G^\prime$ is a proper subtree of $T_s$.
			From the choice of $e_s$ and $e_t$, we know the subgraph represented by the resulting edge subset $(E_s \cup \{ e_t \}) \setminus \{ e_s \}$ forms a tree; let $E_k = (E_s \cup \{ e_t \}) \setminus \{ e_s \}$.
			Furthermore, since $E_k \cap E_t$ includes $E^\prime \cup \{ e_t \}$ and the subgraph formed by $E^\prime \cup \{ e_t \}$ is connected, the subgraph formed by $E_k \cap E_t$ has a connected component whose edge set has size at least $|E^\prime| +1$.
			Therefore, we can conclude that $E_k$ is reconfigurable into $E_t$ by the induction hypothesis.
		\end{proof}

\section{Induced and spanning variants} \label{sec:vertex}
		In this section, we deal with the induced and spanning variants where subgraphs are represented as vertex subsets. 
		Most of our results for these variants are hardness results, except for Theorems~\ref{the:spanningtreetj} and \ref{the:spanningbicliquesolve}.
	
	\subsection{Path and cycle}
		In this subsection, we show that both induced and spanning variants under the TJ or TS rules are PSPACE-complete for the properties ``a graph is a path'' and ``a graph is a cycle.''
		All proofs in this subsection make use of reductions that employ almost identical constructions.
		Therefore, we describe the detailed proof for only one case, and give proof sketches for the other cases.

We give polynomial-time reductions from the {\sc shortest path reconfiguration} problem, which can be seen as a {\sc subgraph reconfiguration} problem, defined as follows~\cite{B13}. 
		For a simple, unweighted, and undirected graph $G$ and two distinct vertices $s,t$ of $G$, {\sc shortest path reconfiguration} is the induced (or spanning) variant of {\sc subgraph reconfiguration} under the TJ rule for the property ``a graph is a shortest $st$-path.'' 
		Notice that there is no difference between the induced variant and the spanning variant for this property, because any shortest path in a simple graph forms an induced subgraph. 
		This problem is known to be PSPACE-complete~\cite{B13}. 
		
		Let $d$ be the (shortest) distance from $s$ to $t$ in $G$.
		For each $i \in \{0,1,\ldots,d\}$, we denote by $L_i \subseteq V(G)$ the set of vertices that lie on a shortest $st$-path at distance $i$ from $s$.
		Therefore, we have $L_0 = \{ s \}$ and $L_d = \{ t \}$.
		We call each $L_i$ a  {\em layer.}
		Observe that any shortest $st$-path contains exactly one vertex from each layer, and 
we can assume without loss of generality that $G$ has no vertex which does not belong to any layer.
%
		

%

		We first give the following theorem.
		\begin{theorem}\label{the:path_tj}
		For the property ``a graph is a path,'' the induced and spanning variants of {\sc subgraph reconfiguration} under the TJ rule are both PSPACE-complete on bipartite graphs. 
		\end{theorem}
		\begin{proof}
		Observe that these variants are in PSPACE. 
		Therefore, 
we construct a polynomial-time reduction from {\sc shortest path reconfiguration}.

		\begin{figure}[t]
			\centering
			\includegraphics[width=0.8\linewidth]{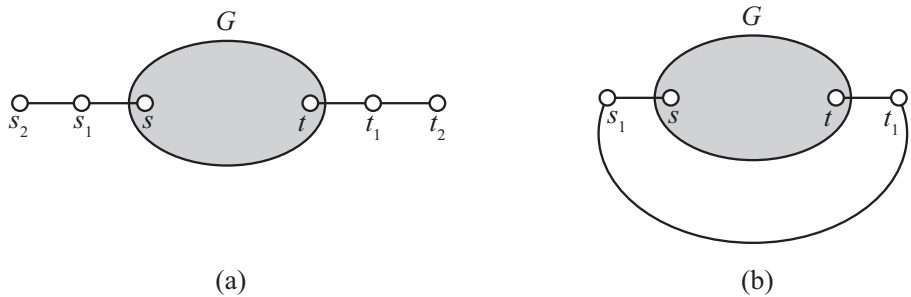}
			\caption{Reductions for the properties (a) ``a graph is a path,'' and (b) ``a graph is a cycle.''}
			\label{fig:pathcycle_tj}
		\end{figure}

		Let $(G,V_s,V_t)$ be an instance of {\sc shortest path reconfiguration}.
		Since any shortest $st$-path contains exactly one vertex from each layer, we can assume without loss of generality that $G$ has no edge joining two vertices in the same layer, that is, each layer $L_i$ forms an independent set in $G$.
		Then, $G$ is a bipartite graph. 
		From $(G,V_s,V_t)$, we construct a corresponding instance $(G^\prime,V^\prime_s,V^\prime_t)$ for the induced and spanning variants; 
note that we use the same reduction for both variants. 
		Let $G^\prime$ be the graph obtained from $G$ by adding four new vertices $s_1,s_2,t_1,t_2$ which are connected with four new edges $s_2s_1$, $s_1s$, $tt_1$, $t_1t_2$. (See \figurename~\ref{fig:pathcycle_tj}(a).)
		Note that $G^\prime$ is also bipartite. 
		We then set $V^\prime_s = V_s \cup \{ s_1,s_2,t_1,t_2\}$ and $V^\prime_t = V_t \cup \{ s_1,s_2,t_1,t_2\}$.
		Since each of $V_s$ and $V_t$ induces a shortest $st$-path in $G$, each of $V^\prime_s$ and $V^\prime_t$ is a feasible solution to both variants.
		This completes the polynomial-time construction of the corresponding instance.
		
 		We now give the key lemma for proving the correctness of our reduction.
		\begin{lemma} \label{lem:path_tj_pro}
			Let $V^\prime \subseteq V(G^\prime)$ be any solution for the induced or spanning variant which is reachable by a reconfiguration sequence from $V^\prime_s$ {\rm (}or $V^\prime_t${\rm )} under the TJ rule.
			Then, $V^\prime$ satisfies the following two conditions{\rm :}
			\begin{listing}{\rm (a)}
				\item[{\rm (a)}] $s_2,s_1,s,t,t_1,t_2 \in V^\prime${\rm ;} and
				\item[{\rm (b)}] $V^\prime$ contains exactly one vertex from each layer of $G$.
			\end{listing}
		\end{lemma}
		\begin{proof}
			We first prove that condition (a) is satisfied.
			Notice that both $V^\prime_s$ and $V^\prime_t$ satisfy this condition.
			Because $V^\prime$ is reconfigurable from $V^\prime_s$ or $V^\prime_t$, it suffices to show that we cannot move any of $s_2,s_1,s,t,t_1,t_2$ in a reconfiguration sequence starting from $V^\prime_s$ or $V^\prime_t$. 
			Suppose for sake of contradiction that we can move at least one of $s_2,s_1,s,t,t_1,t_2$.
			Then, the first removed vertex $v \in \{s_2,s_1,s,t,t_1,t_2\}$ must be either $s_2$ or $t_2$; otherwise the resulting subgraph would be disconnected. 
			Let $w$ be the vertex with which we exchanged $v$.
			Then, $w \in V(G^\prime) \setminus \{s_2,s_1,s,t,t_1,t_2\} = V(G) \setminus \{s,t\}$.
			Therefore, the resulting vertex subset cannot induce a path, and hence it cannot be a solution for the induced variant.
			Similarly, the induced subgraph cannot contain a spanning path, and hence it cannot be a solution for the spanning variant.

			We next show that condition (b) is satisfied.
			Recall that $d$ denotes the number of edges in a shortest $st$-path in $G$.
			Then, we have $|V^\prime| = |V^\prime_s| = |V^\prime_t| = d+5$.
			By condition~(a), we know that $V^\prime$ contains $s_2,s_1,t_1,t_2$, and hence in both induced and spanning variants, $s$ and $t$ must be connected by a path formed by $d+1$ vertices in $V^\prime \setminus \{s_2,s_1,t_1,t_2\}$. 
			Since the length of this $st$-path is $d$, this path is shortest and hence $V^\prime$ must contain exactly one vertex from each of $d+1$ layers of $G$.
		\end{proof}

		Consider any vertex subset $V^\dprime \subseteq V(G^\prime)$ which satisfies conditions~(a) and (b) of Lemma~\ref{lem:path_tj_pro};
note that $V^\dprime$ is not necessarily a feasible solution. 
		Then, these conditions ensure that $V^\dprime \setminus \{ s_2,s_1,t_1,t_2 \}$ forms a shortest $st$-path in $G$ if and only if the subgraph represented by $V^\dprime$ induces a path in $G^\prime$.
		Thus, an instance $(G,V_s,V_t)$ of {\sc shortest path reconfiguration} is a $\YES$-instance if and only if the corresponding instance $(G^\prime,V^\prime_s,V^\prime_t)$ of the induced or spanning variant is a $\YES$-instance.
		This completes the proof of Theorem~\ref{the:path_tj}.
\end{proof}

		Similar arguments give the following theorem. 
		\begin{theorem}\label{the:cycle_tj}
		Both the induced and spanning variants of {\sc subgraph reconfiguration} under the TJ rule are PSPACE-complete for the property ``a graph is a cycle.''
		\end{theorem}
		\begin{proof}
		Our reduction is the same as in the proof of Theorem~\ref{the:path_tj} except for the following point (see also \figurename~\ref{fig:pathcycle_tj}(b)): 
instead of adding four new vertices, we connect $s$ and $t$ by a path of length three with two new vertices $s_1$ and $t_1$. 
		Then, the same arguments hold as the proof of Theorem~\ref{the:path_tj}.
		\end{proof}
		\subsection{Path and cycle under the TS rule}                    
		
		We now consider the TS rule. 
		Notice that, in the proofs of Theorems~\ref{the:path_tj} and \ref{the:cycle_tj}, we exchange only vertices contained in the same layer. 
		Since any shortest $st$-path in a graph $G$ contains exactly one vertex from each layer, we can assume without loss of generality that each layer $L_i$ of $G$ forms a clique.
		Then, the same reductions work for the TS rule, and we obtain the following theorem. 
		\begin{theorem}\label{the:pathcycle_ts}
		Both the induced and spanning variants of {\sc subgraph reconfiguration} under the TS rule are PSPACE-complete for the properties ``a graph is a path'' and ``a graph is a cycle.''
		\end{theorem}

	\subsection{Tree}
		Wasa et al.~\cite{WYA16} showed that the induced variant under the TJ and TS rules is PSPACE-complete for the property ``a graph is a tree.''
		In this subsection, we show that the spanning variant for this property is also PSPACE-complete under the TS rule, while it is linear-time solvable under the TJ rule. 
		
		We first note that our proof of Theorem~\ref{the:pathcycle_ts} yields the following theorem. 
		\begin{theorem}\label{the:spanningtreets}
		The spanning variant of {\sc subgraph reconfiguration} under the TS rule is PSPACE-complete for the property ``a graph is a tree.''
		\end{theorem}
		\begin{proof}
		We claim that the same reduction as in Theorem~\ref{the:pathcycle_ts} applies. 
		Let $V^\prime \subseteq V(G^\prime)$ be any solution which is reachable by a reconfiguration sequence from $V^\prime_s$ (or $V^\prime_t$) under the TS rule, where $(G^\prime,V_s^\prime,V_t^\prime)$ is the corresponding instance for the spanning variant, as in the reduction.
		Then, the TS rule ensures that $s_2,s_1,s,t,t_1,t_2 \in V^\prime$ holds, and $V^\prime$ contains exactly one vertex from each layer of $G$.
		Therefore, any solution forms a path even for the property ``a graph is a tree,'' and hence the theorem follows. 
		\end{proof}

		In contrast to Theorem~\ref{the:spanningtreets}, the spanning variant under the TJ rule is solvable in linear time.
		We note that the reduction in Theorem~\ref{the:spanningtreets} does not work under the TJ rule, because the tokens on $s_2$ and $t_2$ can move (jump) and hence there is no guarantee that a solution forms a path for the property ``a graph is a tree.''
		\begin{theorem}\label{the:spanningtreetj}
		The spanning variant of {\sc subgraph reconfiguration} under the TJ rule can be solved in linear time for the property ``a graph is a tree.''
		\end{theorem}

		Suppose that $(G,V_s,V_t)$ is a given instance.
		We assume that $|V_s| = |V_t| \geq 2$ holds; otherwise it is a trivial instance.
		Then, Theorem~\ref{the:spanningtreetj} can be obtained from the following lemma. 
		\begin{lemma} \label{lem:spanningtreetj}
			$(G,V_s,V_t)$ with $|V_s| = |V_t| \geq 2$ is a $\YES$-instance if and only if $V_s$ and $V_t$ are contained in the same connected component of $G$.
		\end{lemma}
		\begin{proof}
			We first prove the only-if direction. 
			Since $|V_s| = |V_t| \geq 2$ and the property requires subgraphs to be connected, the subgraph induced by any feasible solution contains at least one edge. 
			Because we can exchange only one vertex at a time and the resulting subgraph must retain connected, $V_s$ and $V_t$ are contained in the same connected component of $G$ if $(G,V_s,V_t)$ is a $\YES$-instance. 
			
			Next, we prove the if direction for the case where $|V_s| = |V_t| = 2$.
			In this case, $G[V_s]$ and $G[V_t]$ consist of single edges, say $e_s$ and $e_t$, respectively. 
			Since $V_s$ and $V_t$ are contained in the same connected component of $G$, there is a path in $G$ between $e_s$ and $e_t$. 
			Thus, we can exchange vertices along the path, and obtain a reconfiguration sequence from $V_s$ to $V_t$. 
			In this way, the if direction holds for this case. 
			
			Finally, we prove the if direction for the remaining case, that is, $|V_s| = |V_t| \ge 3$. 
			Consider any spanning trees $T_s$ of $G[V_s]$ and $T_t$ of $G[V_t]$. 
			Since $|V_s| = |V_t| \ge 3$, each of $T_s$ and $T_t$ has at least two edges. 
			Then, if we regard $(G, E(T_s), E(T_t))$ as an instance of the edge variant under the TJ rule for the property ``a graph is a tree,'' we know from the proof of Theorem~\ref{the:edgetree} that it is a $\YES$-instance. 
			Thus, there exists a reconfiguration sequence $\mathcal{E} = \seq{E(T_s)=E_0,E_1,\ldots,E_\ell=E(T_t)}$ of edge subsets under the TJ rule.
			We below show that, based on $\mathcal{E}$, we can construct a reconfiguration sequence between $V_s$ and $V_t$ for the spanning variant under the TJ rule. 
			
			For each $E_i$ in $\mathcal{E}$, let $V_i$ be the vertex set of the tree represented by $E_i$.
			Notice that $V_i$ is a feasible solution for the spanning variant, and that $V_0=V_s$ and $V_\ell=V_t$ hold. 
			We claim that the sequence $\seq{V_s=V_0,V_1,\ldots,V_\ell=V_t}$ of vertex subsets is a reconfiguration sequence for the spanning variant under the TJ rule (after removing redundant vertex subsets if needed). 
			To show this, it suffices to prove $|V_i \setminus V_{i-1}| = |V_{i-1} \setminus V_i| \leq 1$ for all $i \in \{1,2,\ldots,\ell\}$.
			%
			%
			%
			Suppose for the sake of contradiction that there exists $V_i$ such that $|V_i \setminus V_{i-1}| \geq 2$ holds. (See also \figurename~\ref{fig:tree_2vertices}.)
			Since $|E_i| \geq 2$ and $|E_i \setminus E_{i-1}| = 1$ hold, we have $E_i \cap E_{i-1} \neq \emptyset$ and hence $V_i \cap V_{i-1} \neq \emptyset$. 
			Then there is at least one edge $e=uv$ in $E_i \setminus E_{i-1}$ joining a vertex $u \in V_i \setminus V_{i-1}$ and $v \in V_i \cap V_{i-1}$, because $E_i$ must form a connected subgraph.
			Since $|V_i \setminus V_{i-1}| \geq 2$, there is another vertex $u^\prime \neq u$ in $V_i \setminus V_{i-1}$, and there is an edge $e^\prime$ incident to $u^\prime$.
			Note that $e \neq e^\prime$.
			Furthermore, we know that $e^\prime \in E_i \setminus E_{i-1}$ because $u^\prime \in V_i \setminus V_{i-1}$.
			Therefore, we have $e, e^\prime \in E_i \setminus E_{i-1}$, which contradicts the fact that $|E_i \setminus E_{i-1}| = 1$ holds.
		\end{proof}
		
		\begin{figure}[t]
			\centering
			\includegraphics[width=0.3\linewidth]{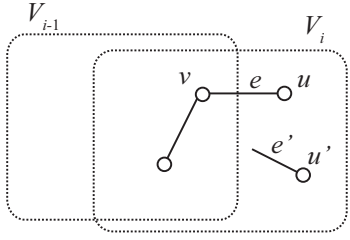}
			\caption{Illustration for the proof of Lemma~\ref{lem:spanningtreetj}.}
			\label{fig:tree_2vertices}
		\end{figure}

	\subsection{Biclique}
		For the property ``a graph is an $(i,j)$-biclique,'' we show that the induced variant under the TJ rule is PSPACE-complete even if $i=j$ holds, or $i$ is fixed.
		On the other hand, the spanning variant under the TJ rule is NP-hard even if $i = j$ holds, while it is polynomial-time solvable when $i$ is fixed.
		


		We first give the following theorem for a fixed $i \ge 1$.
		\begin{theorem} \label{the:inducedbiclique}
		For the property ``a graph is an $(i,j)$-biclique,'' the induced variant of {\sc subgraph reconfiguration} under the TJ rule is PSPACE-complete even for any fixed integer $i \ge 1$.
		\end{theorem}
		\begin{proof}
			We give a polynomial-time reduction from the {\sc maximum independent set reconfiguration} problem~\cite{W14}, which can be seen as a {\sc subgraph reconfiguration} problem.
			The {\sc maximum independent set reconfiguration} problem is the induced variant for the property ``a graph is edgeless'' such that two given independent sets are maximum. 
			Note that, because we are given maximum independent sets, there is no difference between the TJ and TS rules for this problem. 
			This problem is known to be PSPACE-complete~\cite{W14}.
			
			\begin{figure}[t]
				\centering
				\includegraphics[width=0.3\linewidth]{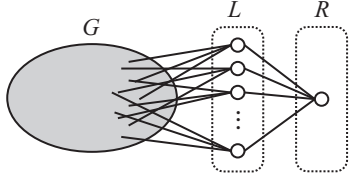}
				\caption{Reduction for the property ``a graph is an $(i,j)$-biclique'' for any fixed $i \ge 1$.}
				\label{fig:biclique_ind}
			\end{figure}
			
			Suppose that $(G,V_s,V_t)$ is an instance of {\sc maximum independent set reconfiguration}.
			We now construct a corresponding instance $(G^\prime,V^\prime_s,V^\prime_t)$ of the induced variant under the TJ rule for the property ``a graph is an $(i,j)$-biclique,'' where $i$ is any fixed positive integer. 
			(See also \figurename~\ref{fig:biclique_ind}.)
			Let $L$ and $R$ be distinct sets of new vertices such that $|L| = i$ and $|R| = 1$.
			The vertex set of $G^\prime$ is defined as $V(G^\prime) = V(G) \cup L \cup R$, and the edge set of $G^\prime$ as $E(G^\prime) = E(G) \cup \{ uv \mid u \in V(G), v \in L \} \cup \{ vw \mid v \in L, w \in R \}$, that is, new edges are added so that there are edges between each vertex of $L$ and  each vertex of $V(G) \cup R$.
			Let $V^\prime_s = V_s \cup L \cup R$ and $V^\prime_t = V_t \cup L \cup R$.
			Since $L$, $R$, $V_s$ and $V_t$ are all independent sets in $G^\prime$, both $V^\prime_s$ and $V^\prime_t$ form $(i,j)$-bicliques, where $i = |L|$ and $j = |V_s \cup R| = |V_t \cup R|$.  We have now completed the construction of our corresponding instance, which can be accomplished in polynomial time. 
			
			Because each vertex in $L$ is connected to all vertices in $V(G)$, a vertex subset $V^\dprime \subseteq V(G^\prime)$ cannot form a bipartite graph (and hence an $(i,j)$-biclique) if $V^\dprime \cap L \neq \emptyset$ or if $V^\dprime$ contains two vertices joined by an edge in $G$.  
			In addition, we cannot move any token placed on $L \cup R$ onto a vertex in $V(G)$ because both $V_s$ and $V_t$ are maximum independent sets of $G$.
			Note that, in the case of $i=1$, there may exist a vertex $u$ in $V(G)$ which is adjacent to all vertices in $V_s$ or in $V_t$. 
			However, $u$ is not adjacent to the vertex in $R$, and hence the token placed on the vertex in $L$ cannot be moved to $u$ in this case, either.  
			Therefore, for any feasible solution $V^\prime \subseteq V(G^\prime)$ which is reconfigurable from $V^\prime_s$ or $V^\prime_t$ under the TJ rule, the vertex subset $V^\prime \cap V(G)$ forms a maximum independent set of $G$. 
			Thus, an instance $(G,V_s,V_t)$ of {\sc maximum independent set reconfiguration} is a $\YES$-instance if and only if our corresponding instance $(G^\prime,V^\prime_s,V^\prime_t)$ is a $\YES$-instance.
		\end{proof}

		The corresponding instance $(G^\prime,V^\prime_s,V^\prime_t)$ constructed in the proof of Theorem~\ref{the:inducedbiclique} satisfies $i=j$ if we set $i = |V_s| + 1 = |V_t| + 1$.
		Therefore, we can obtain the following corollary.
		\begin{corollary}\label{col:inducedbiclique}
		For the property ``a graph is an $(i,j)$-biclique,'' the induced variant of {\sc subgraph reconfiguration} under the TJ rule is PSPACE-complete even if $i=j$ holds.
		\end{corollary}
		
		
		We next give the following theorem.
		\begin{theorem} \label{the:spanningbiclique}
		For the property ``a graph is an $(i,j)$-biclique,'' the spanning variant of {\sc subgraph reconfiguration} under the TJ rule is NP-hard even if $i=j$ holds.
		\end{theorem}
		\begin{proof}
			We give a polynomial-time reduction from the {\sc balanced complete bipartite subgraph} problem, defined as follows~\cite{GJ79}.
			Given a bipartite graph $G$ and a positive integer $k$, the {\sc balanced complete bipartite subgraph} problem is to determine whether or not $G$ contains a $(k,k)$-biclique as a subgraph; this problem is known to be NP-hard~\cite{GJ79}.
			
			Suppose that $(G,k)$ is an instance of {\sc balanced complete bipartite subgraph}, where $G$ is a bipartite graph with bipartition $(A,B)$.
			Then, we construct a corresponding instance $(G^\prime,V_s,V_t)$ of the spanning variant under the TJ rule for the property ``a graph is a $(k,k)$-biclique.''
			We first construct a graph $G^\prime$. (See \figurename~\ref{fig:biclique_span}.) 
			We add to $G$ two new $(k,k)$-bicliques $G_1$ and $G_2$; let $(A_1, B_1)$ be the bipartition of $G_1$, and $(A_2,B_2)$ be that of $G_2$.
			We then add edges between any two vertices $x \in B_1$ and $y \in A$, and between any two vertices $x \in B$ and $y \in A_2$.
			Therefore, $G^\prime[B_1 \cup A]$ and $G^\prime[B \cup A_2]$ are bicliques in $G^\prime$.
			This completes the construction of $G^\prime$.
			We then set $V_s = V(G_1)$ and $V_t = V(G_2)$.
			Then, $V_s$ and $V_t$ are solutions, since $G[V_s]$ and $G[V_t]$ contain $(k,k)$-bicliques $G_1$ and $G_2$, respectively.
			In this way, 
			the corresponding instance can be constructed in polynomial time.
			
			By the construction of $G^\prime$, any reconfiguration sequence between $V_s = V(G_1)$ and $V_t = V(G_2)$ must pass through a $(k,k)$-biclique of $G$. 
			Therefore, the theorem follows. 
		\end{proof}
		\begin{figure}[t]
			\centering
			\includegraphics[width=0.50\linewidth]{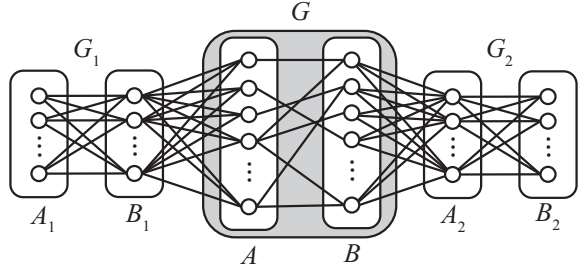}
			\caption{Reduction for the property ``a graph is a $(k,k)$-biclique.''}
			\label{fig:biclique_span}
		\end{figure}

		We now give a polynomial-time algorithm solving the spanning variant for a fixed constant $i \ge 1$.
		\begin{theorem} \label{the:spanningbicliquesolve}
			For the property ``a graph is an $(i,j)$-biclique,'' the spanning variant of {\sc subgraph reconfiguration} under the TJ rule is solvable in polynomial time when $i \ge 1$ is a fixed constant.
		\end{theorem}
		\begin{proof}
			We give such an algorithm. 
			We assume without loss of generality that $j > i$; otherwise both $i$ and $j$ are fixed constants, and hence such a case can be solved in polynomial time by Theorem~\ref{the:xp}.
			
			We will refer to the $i$ vertices in the bounded-size part of the biclique as {\em hubs}, and the $j$ vertices in the other part as {\em terminals}.  
			%
			Let $H \subseteq V(G)$ be an arbitrary vertex subset such that $|H|=i$.
			We denote by $C(H)$ the set of all common neighbors of $H$ in $G$.
			Notice that, if $|C(H)| \ge j$, then the subgraph represented by $H \cup C(H)$ contains at least one $(i,j)$-biclique whose hub set is $H$.
			We denote by $\mathcal{S}(H)$ the set of all solutions that contain $(i,j)$-bicliques with the hub set $H$. 
			We know that $\mathcal{S}(H) = \emptyset$ if $|C(H)| < j$; otherwise $H \cup T$ is in $\mathcal{S}(H)$ for any subset $T \subseteq C(H)$ such that $|T|=j$.
			It should be noted that a solution in the spanning variant is simply a vertex subset $V^\prime$ of $V(G)$, and there is no restriction on how to choose a hub set from $V^\prime$. 
			(For example, if a solution $V^\prime$ induces a clique of size five, then there are ten ways to choose a hub set from $V^\prime$ for $(2,3)$-bicliques.)
			Therefore, 
			$\mathcal{S}(H) \cap \mathcal{S}(H^\prime) \neq \emptyset$ may hold for distinct hub sets $H,H^\prime$.
			
			We describe two key observations in the following.
			The first one is that for a hub set $H$, any two solutions $V_a,V_b \in \mathcal{S}(H)$ are reconfigurable because we can always move vertices in $V_a \setminus V_b$ into ones in $V_b \setminus V_a$ one by one.
			The second one is that for any two distinct hub sets $H_a$ and $H_b$, if there exist $V_a \in \mathcal{S}(H_a)$ and $V_b \in \mathcal{S}(H_b)$ such that $|V_a \setminus V_b| \le 1$ and $|V_b \setminus V_a| \le 1$ (this means that $V_a$ and $V_b$ are reconfigurable by one reconfiguration step, or $V_a = V_b$), then all pairs of solutions in $\mathcal{S}(H_a) \cup \mathcal{S}(H_b)$ are reconfigurable.

			Based on these observations, we construct an {\em auxiliary graph} $A$ for a given instance $(G, V_s, V_t)$, as follows.
			Each node in $A$ corresponds to a set $H$ of $i$ vertices (hubs) in the input graph $G$ such that 
			$|C(H)| \ge j$; we represent the node in $A$ simply by the corresponding hub set $H$.
			Two nodes $H_a$ and $H_b$ are adjacent in $A$ if there exist $V_a \in \mathcal{S}(H_a)$ and $V_b \in \mathcal{S}(H_b)$ such that $|V_a \setminus V_b| \le 1$ and $|V_b \setminus V_a| \le 1$.
			Let $H_s$ and $H_t$ be any two nodes such that $V_s \in \mathcal{S}(H_s)$ and $V_t \in \mathcal{S}(H_t)$, respectively.
			Then, we claim that there is a reconfiguration sequence between $V_s$ and $V_t$ if and only if there is a path in $A$ between $H_s$ and $H_t$.

			We first suppose that there is a path $\mathcal{P}$ in $A$ between $H_s$ and $H_t$. 
			We know that any two consecutive nodes $H_a$ and $H_b$ in $\mathcal{P}$ are adjacent in $A$.
			Then, as we mentioned above, all pairs of solutions in $\mathcal{S}(H_a) \cup \mathcal{S}(H_b)$ are reconfigurable.
			Since $V_s \in \mathcal{S}(H_s)$ and $V_t \in \mathcal{S}(H_t)$, we can thus conclude that there is a reconfiguration sequence between $V_s$ and $V_t$.
			
			We now suppose that there exists a reconfiguration sequence $\mathcal{R}$ between $V_s$ and $V_t$.
			For each solution $V_a$ in $\mathcal{R}$ except for $V_s$ and $V_t$, we choose an arbitrary node $H_a$ in $A$ which satisfies $V_a \in \mathcal{S}(H_a)$.
			Consider any two consecutive solutions $V_a$ and $V_b$ in $\mathcal{R}$.
			Then, by the construction of $A$, the chosen nodes $H_a$ and $H_b$ are adjacent in $A$ (or sometimes $H_a = H_b$) because $|V_a \setminus V_b| =1$ and $|V_b \setminus V_a| = 1$.
			In this way, we can ensure the existence of a desired path in $A$.
			
			The running time of the algorithm depends on the size $i$ of a hub set.  
			Let $n$ be the number of vertices in $G$. 
			The size of the node set of $A$ is in $O(n^i)$.
			For any two nodes $H_a$ and $H_b$ in $A$, we can determine in $O(n)$ time whether there is an edge between them by checking that all of the following four conditions hold or not:
			\begin{listing}{(b)}
				\item[\rm (a)] $|(H_a \cup C(H_a)) \cap (H_b \cup C(H_b))| \geq i+j-1${\rm ;}
				\item[\rm (b)] $|H_a \setminus (H_b \cup C(H_b))| \leq 1${\rm ;}
				\item[\rm (c)] $|H_b \setminus (H_a \cup C(H_a))| \leq 1${\rm ;} and
				\item[\rm (d)] $|H_a \cup H_b| \leq i+j+1$.
			\end{listing}
			Note that, since we have assumed that $i<j$, condition (d) is always satisfied.
			Therefore, we take $O(n^{2i+1})$ time to construct $A$ and to check whether the nodes corresponding to $V_s$ and $V_t$ are connected. 
		\end{proof}

\subsection{Diameter-two graph}\label{sec:diameter}
	In this subsection, we consider the property ``a graph has diameter at most two.''
	Note that the induced and spanning variants are the same for this property. 
	\begin{theorem} \label{the:diamtwo}
	Both induced and spanning variants of {\sc subgraph reconfiguration} under the TS rule are PSPACE-complete for the property ``a graph has diameter at most two.''
	\end{theorem}
	\begin{proof}
		Since the induced variant and the spanning variant are the same for this property, it suffices to show the PSPACE-hardness only for the induced variant. 
		We give a polynomial-time reduction from the {\sc clique reconfiguration} problem, which is the induced variant (also the spanning variant) of {\sc subgraph reconfiguration} for the property ``a graph is a clique.''
		This problem is known to be PSPACE-complete under both the TJ and TS rules~\cite{IOO15}, and we give a reduction from the problem under the TS rule.
		
		\begin{figure}[t]
			\centering
			\includegraphics[width=0.45\linewidth]{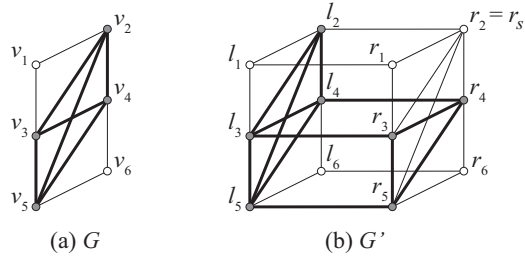}
			\caption{Reduction for the property ``a graph has diameter at most two.'' The vertices of $V_s$ in $G$ and of $V_s^\prime$ in $G^\prime$ are depicted by gray vertices, where $r_s = r_2$.}
			\label{fig:diam2}
		\end{figure}
		
		Suppose that $(G,V_s,V_t)$ is an instance of {\sc clique reconfiguration} under the TS rule such that $|V_s| = |V_t| \geq 2$; otherwise it is a trivial instance.
		Then, we construct a corresponding instance $(G^\prime, V^\prime_s, V^\prime_t)$ of the induced variant under the TS rule. 
		Let $V(G) = \{ v_1,v_2,\ldots,v_n \}$, where $n = |V(G)|$.
		We form $G^\prime$ by making two copies of $G$ and adding edges between corresponding vertices of the two graphs. (See \figurename~\ref{fig:diam2}.) 
		More formally, the vertex set $V(G^\prime)$ is defined as $V(G^\prime) = L \cup R$, where $L = \{  l_i \mid v_i \in V(G) \}$ and $R= \{ r_i \mid v_i \in V(G) \}$, and the edge set $E(G^\prime)$ is defined as $E(G^\prime) = E_l \cup E_r \cup E_c$, where $E_l = \{ l_{i}l_{j} \mid v_{i}v_{j} \in E(G) \}$, $E_r = \{ r_{i}r_{j} \mid v_{i}v_{j} \in E(G) \}$ and $E_c = \{ l_{i}r_{i} \mid v_i \in V(G) \}$.
		For each $i \in \{1,2,\ldots,n\}$, we call $l_i$ and $r_i$ {\em corresponding vertices}, and an edge joining corresponding vertices a {\em connecting edge.}
		We then construct $V^\prime_s$ and $V^\prime_t$.
		We say that a vertex $v^\prime$ in a vertex subset $V^\prime \subseteq V(G^\prime)$ is {\em exposed} in $V^\prime$ if the corresponding vertex of $v^\prime$ does not belong to $V^\prime$.
		We construct $V^\prime_s$ and $V^\prime_t$ so that they each have exactly  one exposed vertex.
		Let $v_s \in V_s$ be an arbitrary vertex in $V_s$ and $v_t \in V_t$ be an arbitrary vertex in $V_t$.
		Then, we let $V^\prime_s = \{ l_i, r_i \mid v_i \in V_s \} \setminus \{ r_s \}$ and $V^\prime_t = \{ l_i, r_i \mid v_i \in V_t \} \setminus \{ r_t \}$.
		Note that $l_s$ and $l_t$ are the unique exposed vertices in $V^\prime_s$ and $V^\prime_t$, respectively.
		Since $V_s$ and $V_t$ form cliques in $G$, $G^\prime[V^\prime_s]$ and $G^\prime[V^\prime_t]$ have diameter at most two.  We have thus 
		constructed our corresponding instance in polynomial time. 
		
		Let $V^\prime$ be any subset of $V(G^\prime)$, and let $V^\prime_L = V^\prime \cap L$ and $V^\prime_R = V^\prime \cap R$.
		The key observation is that $G[V^\prime]$ has diameter more than two if both $V^\prime_L$ and $V^\prime_R$ contain exposed vertices in $V^\prime$. 
		We below prove that $(G,V_s,V_t)$ of {\sc clique reconfiguration} is a $\YES$-instance if and only if the corresponding instance $(G^\prime,V^\prime_s,V^\prime_t)$ of the induced variant under the TS rule is a $\YES$-instance.
		
		We first prove the only-if direction, supposing that
		there exists a reconfiguration sequence $\seq{V_s = V_0,V_1,\ldots,V_\ell=V_t}$ of cliques in $G$.
		Then, we show that $V^\prime_s$ is reconfigurable into $V^\prime_t$ by induction on $\ell$.
		If $\ell = 0$ and hence $V_s = V_t$, then we can obtain $V^\prime_t$ from $V^\prime_s$ by exchanging $r_t$ in $V^\prime_s$ with $r_s$ (or $V^\prime_s = V^\prime_t$ already holds).
		We then consider the case where $\ell \geq 1$.
		Let $\{ v_i \} = V_0 \setminus V_1$ and $\{ v_j \} = V_1 \setminus V_0$; note that $v_i$ and $v_j$ are adjacent in $G$ since we consider the TS rule.
		Then, we consider two vertex subsets $V^\dprime_0 = ( V^\prime_s \cup \{ r_s \} ) \setminus \{ r_i \}$ and $V^\prime_1 = ( V^\dprime_0 \cup \{ l_j \} ) \setminus \{ l_i \}$; note that $r_s$ and $r_i$ are adjacent in $G^\prime$ since $r_s, r_i \in V_0 = V_s$, and that $l_i$ and $l_j$ are adjacent in $G^\prime$ since $v_i$ and $v_j$ are adjacent in $G$.
		Notice that $V^\dprime$ and $V^\prime_1$ have distinct exposed vertices $l_i$ and $l_j$, respectively.
		By the construction of $G^\prime$, since $V_0$ and $V_1$ are cliques in $G$,  both $G^\prime[V^\dprime_0]$ and $G^\prime[V^\prime_1]$ have diameter at most two. 
		Then, the sequence $\seq{V^\prime_s,V^\dprime_0,V^\prime_1}$ is a reconfiguration sequence between $V^\prime_s$ and $V^\prime_1$.
		Therefore, by applying the induction hypothesis to $(G^\prime, V^\prime_1, V^\prime_t)$ and $(G,V_1,V_t)$, we obtain a reconfiguration sequence between $V^\prime_1$ and $V^\prime_t$. 
		Thus, we can conclude that $(G^\prime, V^\prime_s, V^\prime_t)$ is a $\YES$-instance.
		
		To prove the if direction, we now suppose that 
		there exists a reconfiguration sequence $\mathcal{V}^\prime = \seq{V^\prime_s=V^\prime_0,V^\prime_1,\ldots,V^\prime_\ell=V^\prime_t}$.
		Consider the case where $V^\prime_i \in \mathcal{V}^\prime$ has the exposed vertex in the side $L$, say $l_i \in V^\prime_i$, and hence $r_i \not\in V^\prime_i$; the other case is symmetric. 
		Because $G^\prime[V^\prime_{i+1}]$ must have diameter at most two (and hence it must have only one exposed vertex), we know that $V^\prime_{i+1}$ is obtained from $V^\prime_i$ by one of the following three moves 
		(1) a token on a vertex $r_k \in V^\prime_i$ is moved to $r_i$;
		(2) the token on $l_i$ is moved to its corresponding vertex $r_i$; or
		(3) the token on $l_i$ is moved to a vertex $l_j \in L$ which is not in $V^\prime_i$ and is adjacent to all vertices $V^\prime_i \cap L$.
		Notice that the other moves increase the number of exposed vertices or make the resulting graph have diameter more than two.
		Then, each $V^\prime_i \in \mathcal{V}^\prime$ induces a clique of size $|V_s| = |V_t|$ in either $L$ or $R$, and we can obtain a desired sequence of cliques between $V_s$ and $V_t$.  
	\end{proof}

	We note that the TS rule is critical in the reduction of Theorem~\ref{the:diamtwo}. 
	Under the TJ rule, 
there is no guarantee that we maintain a clique (and we cannot even guarantee that the resulting clique gets bigger). 

	\section{Conclusions and future work}

	The work in this paper initiates a systematic study of {\sc subgraph reconfiguration}.  
	Although we have identified graph structure properties which are harder for the induced variant than the spanning variant, it remains to be seen whether this pattern holds in general.  
	For the general case, questions of the roles of diameter and the number of subgraphs satisfying the property are worthy of further investigation.
%
	Another obvious direction for further research is an investigation into the fixed-parameter complexity of {\sc subgraph reconfiguration}.

	A natural extension of {\sc subgraph reconfiguration} is the extension from isomorphism of graph structure properties to other mappings, such as topological minors.  

\bibliographystyle{abbrv}
\bibliography{subgraph}


\end{document}